\documentclass[11pt,letterpaper]{article}

\usepackage[margin=1in]{geometry}

\usepackage{amsmath,amssymb,amsthm}

\usepackage[noblocks]{authblk}

\usepackage{graphicx}
\usepackage{booktabs} %

\usepackage{algorithm}
\usepackage[noend]{algorithmic}

\usepackage{hyperref}
\hypersetup{
    colorlinks=true,
    linkcolor=blue,
    citecolor=green,
    urlcolor=magenta,
    pdftitle={Large-Scale Graph Building in Dynamic Environments: Low Latency and High Quality},
    pdfauthor={Almeida, et al.}
}

\usepackage[numbers,sort&compress]{natbib}
\usepackage{microtype}
\usepackage{graphicx}
\usepackage{caption}
\usepackage{subcaption}
\usepackage{booktabs} %
\usepackage[normalem]{ulem}

\usepackage{tikz} 
\usepackage{pgfplots}
\pgfplotsset{compat=1.14}
\usepgfplotslibrary{fillbetween}
\usepackage{hyphenat}
\usepackage{subcaption}
\usepackage{float}

\usepackage{graphicx}
\graphicspath{{figures/}}
\def\input@path{{figures/}}

\usepackage[inline]{enumitem}
\newlist{inenum}{enumerate*}{1}
\setlist[inenum]{label=(\roman*)}
\usepackage{amssymb}
\usepackage{amsmath}
\usepackage{nicefrac}
\usepackage{graphicx}
\usepackage{tikz,pgfplots}

\usepackage{algorithm}
\usepackage[noend]{algorithmic}

\usepackage{bbm}
\usepackage{adjustbox}
\usepackage{multirow}
\usepackage{xspace}

\usepackage{thm-restate}

\usepackage{natbib}
\usepackage{xcolor}
\usepackage{xspace}

\usepackage{amsthm}
\usepackage{thmtools} 
\usepackage{thm-restate}

\newtheorem{theorem}{Theorem}[section]
\newtheorem{lemma}[theorem]{Lemma}

\usepackage[capitalize, nameinlink]{cleveref}
\crefname{theorem}{Theorem}{Theorems}
\Crefname{lemma}{Lemma}{Lemmas}
\Crefname{claim}{Claim}{Claims}
\Crefname{observation}{Observation}{Observations}
\Crefname{invariant}{Invariant}{Invariants}

\newcommand{\dist}{\textsc{Dist}\xspace}

\newcommand{\cM}{\mathcal{M}}

\newcommand{\sol}[1]{}

\newcommand{\appendixref}{\if\fullversion1 \cref{appendix:extra-experiments}\xspace \else the Appendix\xspace \fi}

\definecolor{armygreen}{rgb}{0.29, 0.33, 0.13}
\definecolor{darkolivegreen}{rgb}{0.33, 0.42, 0.18}
\definecolor{ao(english)}{rgb}{0.0, 0.5, 0.0}

\usepackage{todonotes}

\begin{document}

\def\gplay{1}

\title{\textbf{Large-Scale Graph Building in Dynamic Environments: \\ Low Latency and High Quality}}

\author[1]{Filipe Miguel Gonçalves de Almeida}
\author[1]{CJ Carey}
\author[1]{Hendrik Fichtenberger}
\author[1]{Jonathan Halcrow}
\author[1]{Silvio Lattanzi}
\author[1]{André Linhares}
\author[1]{Tao Meng}
\author[1]{Ashkan Norouzi-Fard}
\author[1]{Nikos Parotsidis}
\author[1]{Bryan Perozzi}
\author[1]{David Simcha}

\affil[1]{Google}
\affil[ ]{\texttt{\{filipea, cjcarey, fichtenberger, halcrow, silviol, linhares, taomeng, ashkannorouzi, nikosp, hubris, dsimcha\}@google.com}}

\date{}

\maketitle

\begin{abstract}

Learning and constructing large-scale graphs has attracted attention in recent decades, resulting in a rich literature that introduced various systems, tools, and algorithms. Grale~\cite{grale_paper} is one of such tools that is designed for offline environments and is deployed in more than 50 different industrial settings at Google. %
Grale is widely applicable because of its ability to efficiently learn and construct a graph on datasets with multiple types of features.
However, it is often the case that applications require the underlying data to evolve continuously and rapidly and the updated graph needs to be available with low latency.
Such setting make the use of Grale prohibitive.
While there are Approximate Nearest Neighbor (ANN) systems that handle dynamic updates with low latency, they are mostly limited to similarities over a single embedding.

In this work, we introduce a system that inherits the advantages and the quality of Grale, and maintains a graph construction in a dynamic setting with tens of milliseconds of latency per request.
We call the system \emph{Dynamic Grale Using ScaNN} (Dynamic GUS). 
Our system has a wide range of applications with over 10 deployments at Google. %
\if\gplay1%
One of the applications is in Android Security and Privacy, where Dynamic Grale Using ScaNN enables capturing harmful applications $4$ times faster, before they can reach users.
\fi%
\end{abstract}


\hrulefill %

\vspace{-0.2cm}
\section{Introduction}

Mining real-world data often requires dealing with structured data that can conveniently represented by graphs.
This fact has motivated a surge of interest in graph-based methods for semi-supervised learning which combine the power of sparse labeled data items with structure present across the entire dataset~\cite{DBLP:conf/icml/BlumC01,DBLP:journals/corr/abs-2005-03675,DBLP:conf/icml/YangCS16,DBLP:conf/cikm/WuZA18,DBLP:conf/nips/ZhouBLWS03}. %

One of the challenges in deploying graph-based solutions is that in some cases the network structure is not explicitly defined and should be learned from data. This is particularly important in multimodal applications where it is crucial to be able to combine different sources of information efficiently. In fact in those settings, we have a wealth of different modes to compute similarities, each of which may be suboptimal in some way~\cite{DBLP:conf/pkdd/SousaRB13}. For example, video data comes with visual, audio, and text signals, in addition to other kinds of video metadata---each with their own set of similarity measures. Thanks to its practical importance, the task of efficiently learning similarity measures between objects and construct a graph representing them has been extensively studied in recent years~\cite{murua2008potts, ravi2016large, DBLP:conf/pkdd/SousaRB13, grale_paper}.

In a recent work \cite{grale_paper}, the authors proposed a novel approach to compute a similarity measure by combining the similarities of individual point features together. Their algorithm, named Grale, is motivated by the challenges faced when applying semi-supervised learning to rich heterogeneous datasets in industrial settings. Grale can learn pairwise similarity models and build graphs for datasets with billions of points over the entire input in a matter of hours, but it is not able to operate in a dynamic setting where points are inserted, deleted, and modified in a continuous and rapid stream of mutations.

In today's world, large-scale unlabeled data is generated and modified rapidly. for instance, applications such as detecting policy-violating content in an online social media platform with thousands of uploads per second, or finding similar items in recommendation systems with thousands of new entities per second, it is critical to have the ability to build and maintain a graph in real time without recomputing it from scratch. In such applications, it is equally important that the solution limits data staleness to no more than a few seconds (i.e. the underlying data structure must be updated nearly instantly to reflect new mutations received by the system). Those applications, in addition to the successful deployment of Grale in a wide range of applications at Google, motivated our work on finding a solution for computing graphs similar to Grale in a dynamic setting.

Nearest Neighbors search (NN search), or Approximate Nearest Neighbor search (ANN search), can be used to efficiently identify the most similar data points to each other point (or to a query point, if the queries are made in an online mode), which can be used to generate a graph where each point is connected to each most similar point.
Such approaches and can often handle dynamically evolving datasets. As discussed above, this restrict the applicability of such tools as they exhibit the inherent shortcoming of requiring a single embedding representation of each data point.

In this work we combine the best of both families of approaches that we mentioned, namely: 1) the dynamic nature of ANN search algorithms (such as ScaNN\cite{scann}) on single-embedding data representation, and 2) support of multimodal data representation that is often required by real-world applications.
In a nutshell, our approach is to introduce an efficient transformation algorithm that transforms a multimodal representation of a data point to a sparse embedding representation, such that when the various features of two points are similar to each other then these points are also similar in the sparse embedding space.
We use this transformation to compute a set of candidate nearest neighbors of a point, which we then post-process to determine the accurate similarity of the points w.r.t. their multimodal representation.
As a result, we introduce a novel approach named Dynamic Grale Using ScaNN (Dynamic GUS), that achieves all the requirements for a widely-applicable dynamic graph system operating on multimodal datasets.
Notably, our approach is widely used inside Google, where it acts as a critical component of dynamic graph building systems. The core achievement of Dynamic GUS is to compute high-quality neighborhoods of points with sub-second latency in a dynamic environment.

Dynamic GUS has the following capabilities:
\begin{itemize}
    \item Hundreds of thousands of new points with their respective features can be inserted, modified, or deleted per second.
    \item The neighborhoods of hundreds of thousands of new or previously known points can be queried per second. The answers are provided with sub-second latency and the data freshness is within seconds for the 99th percentile of queries.
    \item The neighborhood for a given point is similar to the one created by Grale, if it is ran from scratch on the entire input.
\end{itemize}

Dynamic GUS is not only a proved valuable tool across Google with over 10 deployments across multiple products and use cases; it is also the backbone of various other dynamic and real-time graph algorithms. Indeed, the computed neighborhood of the query points is used to enable more involved graph mining algorithms, including but not limited to Clustering, Label Propagation, and Graph Neural Networks (GNNs).

\if\gplay1
\vspace{-0.1cm}
\subsection{Example Application}

Google Play is a store with a large number of users and millions of apps. The Android Security and Privacy team is actively developing solutions for protecting the users from harmful apps. They use Dynamic Grale Using ScaNN to utilize the dynamic risk models to detect potentially harmful apps in the store. This system helps to bring in $+40\%$ of action rate improvement, and a factor of 4 reduction in detection latency. Dynamic Grale Using ScaNN enables capturing harmful apps much faster, before they can reach users.

\fi
\vspace{-0.1cm}
\section{Preliminaries}
Quantization based techniques are the current state-of-the-art for scaling dot-product search to massive databases and a wide range of applications (e.g., user query~\cite{DBLP:conf/recsys/CremonesiKT10}, classification tasks~\cite{DBLP:conf/cvpr/DeanRSSVY13}, training tasks~\cite{DBLP:conf/icml/YenKYHKR18,DBLP:conf/icml/MussmannE16}) resulted in a rich line of work with  multiple efficient implementations~\cite{DBLP:conf/cvpr/GeHK013,babenko2014additive,scann,ChenW21,WangL12,WangWZTGL12,DBLP:journals/tbd/JohnsonDJ21,8594636}.

ScaNN (Scalable Nearest Neighbor ~\cite{scann}) provides large-scale, efficient vector search while delivering fast indexing time and low memory footprint. ScaNN incorporates multiple techniques including anisotropic loss for training the multi-level search tree~\cite{scann} and orthogonality-amplified residuals to introduce independent, effective redundancy~\cite{sun2024soar}. 
We use ScaNN as a core component of our approach with a significant role in quality and latency of Dynamic GUS. In short, ScaNN is a very flexible system that computes the nearest neighbors of a given query point in a dynamic setting with a wide range of applications.
It supports various distance measures and operates on both sparse and dense embeddings.\footnote{ScaNN's capabilities are beyond the ones described in~\cite{scann}, and some are currently only available internally at Google.} In this work, we are interested in sparse embeddings and the distance measure that we use is the negative of inner dot product (i.e., dot product times $-1$). 
 
Let us start by introducing some notation before providing more details on ScaNN and how it is used in our system. Given a set of points $P$, a sparse embedding $\cM$ maps each point $p \in P$ to a finite-support real vector denoted by $\cM (p)$. 
The distance between two points $p_1, p_2 \in P$ according to the embedding $\cM$ is defined to be the negative of the dot product between their embeddings:
\vspace{-0.4cm}

\begin{align*}
    \dist_{\cM} (p_1, p_2) = - \cM(p_1) \cdot \cM(p_2)\,.
\end{align*}

We omit the subscript $\cM$ and write simply $\dist (p_1, p_2)$ when  the embedding $\cM$ is clear from the context. One of the functionalities of ScaNN is that it finds the nearest neighbors of a given point $p_0$ (potentially $p_0 \notin P$) in a dynamic setting. More precisely, given $p_0$ and an integer $k$, it finds a set $X \subseteq P$ of size $k$ such that:
\vspace{-0.4cm}

\begin{align*}
     \max \{\dist_{\cM} (p, p_0) : p \in X\} \leq \dist_{\cM} (p', p_0) \quad \forall p' \in P \setminus X \,.
\end{align*}
Furthermore, ScaNN has the capability to find all the points with distance  lower than a specified threshold $\tau_0$ to $p_0$, i.e.,
\vspace{-0.4cm}

\begin{align*}
    X = \{p \in P \mid \dist_{\cM} (p, p_0) \leq \tau_0\} \,.
\end{align*}

Instead of ScaNN, any (approximate) nearest neighbor index that is dynamic and supports high-dimensional sparse vectors could be used. While many target libraries low-dimensional, dense data sets, pgvector~\cite{githubGitHubPgvectorpgvector}, NGT~\cite{githubGitHubYahoojapanNGT} or \cite{githubGitHubKelidciknn,li2016fast,githubGitHubFacebookresearchpysparnn} are examples of libraries that support millions or billions of sparse vectors with at least thousands of dimensions. We note that the performance and quality of our solution when used with another library than ScaNN may differ from what we describe in this paper, though.

\section{Dynamic Grale Using ScaNN}

In this section we describe the Dynamic Grale Using ScaNN system (\emph{Dynamic GUS}) and explain its components from a high level point of view. The details of the components are described in the following sections. Let us start by elaborating its functionality and capabilities.

\subsection{How it is used}
Our system supports two main types of Remote Procedure Calls (RPCs). The first type enables users to modify the underlying data by inserting new points, deleting existing points, and changing point features. We refer to these as Mutation RPCs, and they simply return an acknowledgment that the process has executed successfully.
The second type computes the neighborhood of a specified point, which can be either new or existing. Dynamic GUS returns a response containing a list of all similar points and their similarity scores.
All these operations are executed with sub-second latency in a high-query-per-second (QPS) setting.

\subsection{Components}
Dynamic Grale Using ScaNN has three main components that enable it to process and respond to the aforementioned RPCs.

\paragraph{Embedding Generation} One of the main components of our system is an {\em Embedding Generator} algorithm, which generates sparse embeddings for the points. It is crucial for this component to have a very low latency since it is in the critical path for both mutation and neighborhood computation. To achieve that, it needs to operate with local information (e.g., the features of the point) and cannot afford to use techniques that require a high volume of computation. In this work we design and implement an efficient Embedding Generation algorithm that, given the features of a point, computes a quality sparse embedding for it without using any extra information. Moreover, we describe techniques–namely, filtering and inverse document frequency–that using precomputed data, improve the quality of the Embedding Generator algorithm. We provide more details about our proposed algorithm in the following sections.

\paragraph{Neighbors Computation}
After computing the embedding for a query point, we compute its neighbors using any algorithm that finds the nearest neighbors in a given embedding space. A rich line of research has been dedicated to finding nearest neighbors in an embedding space, and in this work we use ScaNN~\cite{scann}.
\paragraph{Similarity Computation}
We then compute the similarity between two points using a pre-trained model based on their features. Any desired model can be used, e.g., Deep Neural Networks, Decision Trees, and Large Language Models.

\subsection{Dynamic Operations}
Equipped with the components of our system, we are now able to detail how Dynamic GUS handles different RPCs.

\subsubsection{Inserting a new point, or updating an existing point.} When Dynamic GUS receives a request to insert or update a point, it first computes the embedding based on the provided features using the Embedding Generator algorithm. Then, it inserts or updates the point with its embedding in the underlying ScaNN system. This operation typically takes up to a second. Afterward, the inserted point will appear in the neighborhood of any point for which the neighborhood is computed, provided it is close enough. More precisely:

\begin{enumerate}
    \item User sends an RPC call to Dynamic GUS to insert or update the point $p$.
    \item Dynamic GUS computes a sparse embedding using the Embedding Generator, denoted by $\mathcal{M}(p)$, and inserts or updates $(p, \mathcal{M}(p))$ in ScaNN.
    \item Dynamic GUS returns an acknowledgment of the point being inserted to the client.
\end{enumerate}
\cref{fig:insert_or_update} also elaborates on how inserting or updating a new point is executed in Dynamic GUS.
\begin{figure}[t]
    \centering
    \includegraphics[width=8cm]{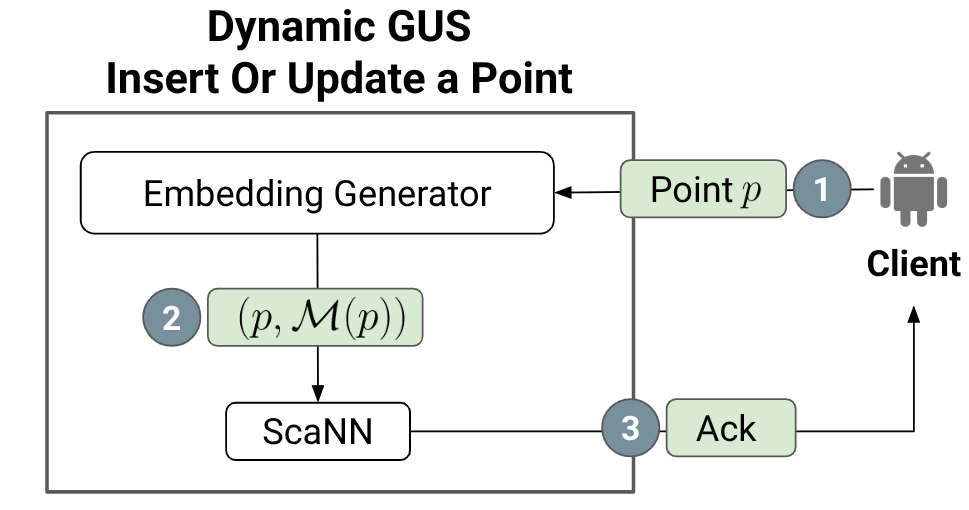}
    \caption{Inserting or Updating a point $p$ to Dynamic GUS.}
    \label{fig:insert_or_update}
\end{figure}

\subsubsection{Deleting a point} Given a request to delete a point, Dynamic GUS simply removes it from ScaNN. This operation is very fast and efficient.

\subsubsection{Computing the neighborhood of a point} Given a request to compute a neighborhood of a point $p$, Dynamic GUS first computes its embedding based on the features of the point using the Embedding Generator algorithm. Then it requests the closest points to $p$ (and their features) from ScaNN, denoted by $Q$. Afterwards it computes the weights of the edges between $p$ and all the points in $Q$ based on the machine-learning model that was trained. More precisely:
\begin{enumerate}
    \item User sends an RPC call to Dynamic GUS to compute the neighborhood of the point $p$.
    \item Dynamic GUS computes a sparse embedding using the Embedding Generator, denoted by $\mathcal{M}(p)$.
    \item Dynamic GUS sends to ScaNN a request to compute the nearest neighbors of $(p, \mathcal{M}(p))$, and obtains as response a set of points $Q$ that are close to $p$ in the embedding space.
    \item Dynamic GUS computes the similarity between $p$ and each point in the set $Q$.
    \item Dynamic GUS returns the neighbors $Q$ to and their similarity $S$ to the user. 
\end{enumerate}

\cref{fig:neighbor} also elaborates on how computing the neighborhood of a point point is executed in Dynamic GUS.
\begin{figure}[t]
    \centering
    \includegraphics[width=8cm]{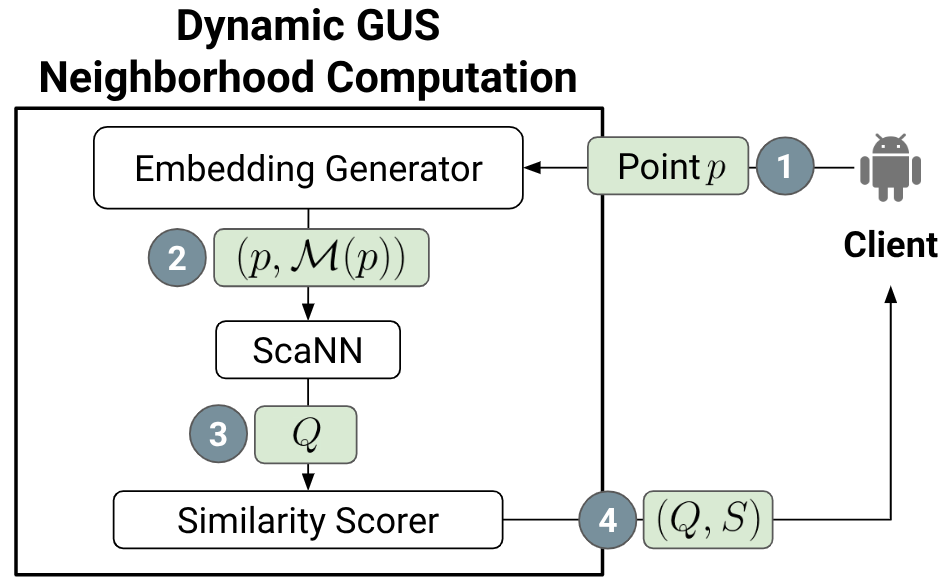}
    \caption{Computing the neighbors of a point $p$ with Dynamic GUS}
    \label{fig:neighbor}
\end{figure}

\section{Technical Details of GUS and its Components}

In this section we describe the different components of Dynamic GUS and provide more details of them. To that end, we start by describing Grale~\cite{grale_paper}. Grale has three main steps:
\begin{enumerate}
    \item \label{enumerate:grale-model} Train a machine-learning model that computes the similarity of two points.
    \item \label{enumerate:grale-sketching} Find \emph{scoring pairs} of points to be scored with the model based on a Locality Sensitive Hashing approach.
    \item \label{enumerate:grale-scoring} Score the edges of each scoring pair with the model.
\end{enumerate}

In the Dynamic GUS system, we assume that the model is trained in an offline setting in the same fashion as in Grale, as described next (we provide more details on how the model can be trained periodically and served in our system in the following sections).
To compute scoring pairs, Grale starts by creating for each point a list of bucket IDs based on its features, via Locality Sensitive Hashing. Intuitively, points that share a bucket ID are likely to be similar and can be considered to be scored. Afterwards, for each bucket ID, a list of all the points that contain this bucket ID is created.
Then all the pairs of points in each bucket are considered to be \emph{scoring pairs}. More details on how these buckets are created can be found in \cite{grale_paper}.%
For the sake of this paper, it is not important how these buckets are created; that can be done via an algorithm completely independent of Grale. Let us provide an example. Assume that we have three points with the following buckets:
\begin{itemize}
    \item Point $p_1$ with buckets: $\{b_1, b_2, b_4\}$
    \item Point $p_2$ with buckets: $\{b_1, b_3\}$
    \item Point $p_3$ with buckets: $\{b_3\}$
\end{itemize}
Then the points in each bucket would be:
\begin{itemize}
    \item Bucket $b_1$ has points: $\{p_1, p_2\}$
    \item Bucket $b_2$ has points: $\{p_1\}$
    \item Bucket $b_3$ has points: $\{p_2, p_3\}$
    \item Bucket $b_4$ has points: $\{p_1\}$
\end{itemize}
Therefore the candidate pairs are $(p_1, p_2)$ and $(p_2, p_3)$, and their similarities will be computed via the model in Grale.

While Grale has proven to be effective across dozens of applications at Google, the original design is not optimized for a dynamic setting. The reason is that even though computing the list of bucket IDs is fast and can be easily implemented in a dynamic setting, keeping the buckets updated and computing the scores of all scoring pairs based on them is very costly.

In Dynamic GUS, we use an approach based on embedding the points in a sparse space; to compute the neighborhood of a point, we retrieve the closest points to it in the embedding space. We provide more details in the remainder of this section.
\vspace{-0.1cm}

\subsection{Sparse Embedding Generation} \label{sec:embedding-construction}
In this section, we present our main algorithm which finds scoring pairs for in a dynamic setting. To that end, we first compute a sparse embedding for each point, which depends only on its features and not on any other points.
There are various properties that we expect from our embedding generator algorithm. First and foremost, for each point $p$, the set of close points to it is expected to be very similar to the set of points that would be considered neighbors of $p$ if we were to run Grale from scratch. Otherwise, we cannot expect to obtain similar quality as Grale. Moreover computing the embedding should admit an efficient implementation in a dynamic setting, to ensure low latency for the system.%

Keeping those points in mind, we start by defining an embedding algorithm to be used in Dynamic GUS, and then provide ideas to improve its quality. 

We define the embedding $\cM$ as follows. Let $\{b_1, b_2, \ldots b_\ell\}$ be the set of bucket IDs that Grale produces for point $p$ (as mentioned before, these buckets can be done via any other algorithm as well). The embedding of point $p$ has $\ell$ non-zero dimensions. These dimensions are $b_1, b_2, \ldots b_\ell$, and their weights are $1.0$ (everywhere else the weight is $0.0$), i.e.,

\begin{table} [htbp]
    \centering
    \begin{tabular}{r|c|c|c|c|}
         Dimensions: & $b_1$ & $b_2$ & $\cdots$ & $b_\ell$ \\
             \hline
         Weights of the dimensions: & $1.0$ & $1.0$ & $\cdots$ & $1.0$ 
    \end{tabular}
\end{table}
\vspace{-0.4cm}

To illustrate this, assume we have the following points:

\begin{itemize}
    \item Point $p_1$ with buckets: $\{b_1, b_2\}$
    \item Point $p_2$ with buckets: $\{b_1\}$
    \item Point $p_3$ with buckets: $\{b_2\}$
\end{itemize}
Then the sparse embeddings for those points are presented in \cref{example:sparse-embedding}.

\begin{table} [htbp]
    \centering
    \begin{tabular}{r|c|c|}
        $P_1$ dimensions: & $b_1$ & $b_2$ \\
             \hline
        $P_1$ weights of the dimensions: & $1.0$ & $1.0$ 
    \end{tabular}
    \centering
    \begin{tabular}{r|c|}
        $P_2$ dimensions: & $b_1$ \\
             \hline
        $P_2$ weights of the dimensions: & $1.0$
    \end{tabular}
    \centering
    \begin{tabular}{r|c|}
        $P_3$ dimensions: & $b_2$ \\
             \hline
        $P_3$ weights of the dimensions: & $1.0$
    \end{tabular}
    \caption{An example of sparse features generated by Dynamic GUS.     \label{example:sparse-embedding}}
\end{table}

Computing the embedding of a point is rather straightforward and in most cases only take a few milliseconds, therefore its impact on latency is negligible. This embedding is produced whenever a new point is inserted and the new point with its embedding is added to ScaNN. When we desire to compute the neighborhood of a point, we compute this embedding again and request ScaNN to return the closest points. We still need to argue that the set of nearest points to $p$ with respect to this embedding is similar to the set of neighbors that Grale produces for the same point. In
~\cref{lemma:same-output} we state this property formally; recall that the distance between two points $p_1, p_2$ in ScaNN is defined to be:
$$\dist_{\cM}  (p_1, p_2)= - \cM(p_1) \cdot \cM(p_2).$$

\begin{lemma} \label{lemma:same-output}
For any point $p$, the neighborhood of $p$ is exactly the same in Grale and Dynamic GUS if we retrieve all the points with negative distance to $p$ in ScaNN.
\end{lemma}

\begin{proof}
Notice that for any two points $p_1, p_2$ their distance in $\cM$ is equal to minus the number of non-zero dimensions they share. Therefore $\dist(p_1, p_2) < 0$ if and only if the points share a non-zero dimension, which happens if and only if they share a bucket ID. By definition of Grale, two points are considered a scoring pair if and only if they share a bucket ID.
\end{proof}

Notice that~\cref{lemma:same-output} holds more generally for any embedding that maps each point $p$ to a vector having strictly positive values in the entries indexed by its bucket IDs (and value zero in the remaining dimensions). This observation guarantees that the lemma holds even after the improvements discussed in the next section.

In most practical applications, we are not interested in retrieving {\em all} the points that have negative distance to a query point $p$ in the embedding $\cM$, for two reasons: \begin{inenum} \item the total number of such points might be very large, which could result in a high latency which may not be suitable for practical use cases; and \item a large fraction of those points might be quite different from the query point (with respect to the similarity measure)\end{inenum}. To address that problem, in practice we bound the number of nearest points retrieved from ScaNN. 
\vspace{-0.1cm}

\subsection{Filtering and Inverse Document Frequency}
In this section we discuss two ideas to improve the embedding; more details on how we apply them to Dynamic GUS are provided in the later sections. 

\paragraph{Filtering overly popular bucket IDs.} Intuitively, if a fraction of the points (e.g., above 20-30\%) have the same bucket ID, this bucket is not a reliable source for finding candidate neighbors in both Grale and Dynamic GUS. That is, in almost all applications only a small number of points are similar to a given query point (typically below 1\% or even 0.01\% of the points). Overly popular buckets can be the result of common features, for instance the word "the" or "a" in a sentence. If a bucket ID is overly popular, we simply ignore it and do not create the corresponding non-zero dimension in the embedding. In the following sections we provide more information on how the overly popular bucket IDs can be computed.

\paragraph{Inverse Document Frequency}
In the embedding defined above, we did not take advantage of the possibility of using weights other than $\{0, 1\}$ for the coordinates of the embedding. Inverse Document Frequency (IDF) is a widely used idea in the literature. It assigns higher importance to entities (for instance words in a document) that are less common. We take advantage of this concept to define weights for the bucket IDs. We define the weight of a bucket ID to be the logarithm of the total number of points over the number of points that have this bucket ID. More precisely, let $P$ be the set of all the points for which we compute the embedding. For any bucket ID $b_i$, let $N(b_i)$ denote the number of points in $P$ that have this bucket ID. The weight of dimension indexed by $b_i$ is defined as
$$\log \left( \frac{|P|}{N(b_i)} \right).$$

For instance, we might create the following embedding for a point.
\vspace{-0.1cm}

\begin{table} [htbp]
    \centering
    \begin{tabular}{r|c|c|c|c|}
         Dimensions: & $b_1$ & $b_7$ & $b_{13}$ & $b_{14}$ \\
             \hline
         Weight of the dimension: & $1.0$ & $2.3$ & $1.4$ & $4.12$ 
    \end{tabular}
\end{table}
\vspace{-0.4cm}

\subsection{Offline preprocessing and periodic reloading}
In all the applications that Dynamic GUS is deployed on, an initial set of points is provided before the start of the Dynamic GUS service. We compute the embeddings of all those points and insert them into ScaNN. We use those points to train a model, identify popular bucket IDs for filtering, and compute the IDF scores to be used in Dynamic GUS. We keep those three components in the memory of the machines to make sure we have fast access to them when the embedding is computed for a point and also when we are computing the similarities for the scoring pairs computed by ScaNN. Notice that we can retrain the model and recompute overly popular bucket IDs and IDF values periodically to ensure that they remain approximately consistent with the evolving dataset.

\section{Empirical Evaluation}

\begin{figure*}
    \centering
    \begin{subfigure}[b]{0.48\textwidth}
        \centering
        \includegraphics[width=\textwidth]{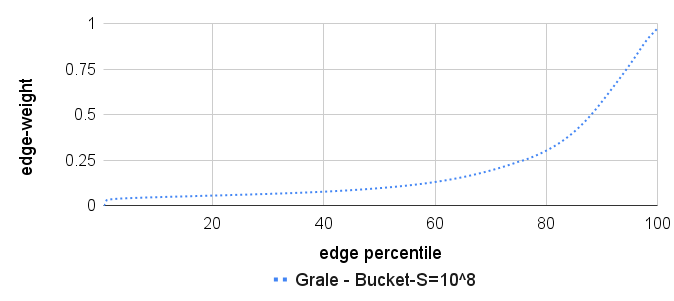}
        \caption{ogbn-products}
    \end{subfigure}
    \hfill
    \begin{subfigure}[b]{0.48\textwidth}
        \centering
        \includegraphics[width=\textwidth]{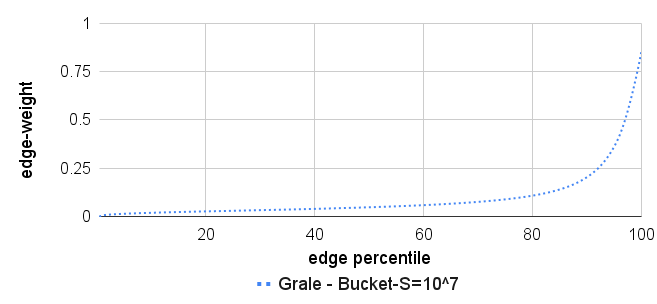}
        \caption{ogbn-arxiv}
    \end{subfigure}
\caption{\label{fig:all-edges} The edge-weight distribution when bucket splitting is not used in Grale and all the points with negative distance are retrieved from ScaNN. The total number of edges retrieved is $175,608,580,162$ for ogbn-products, and $28,515,851,126$ for ogbn-arxiv. The edges are identical for both Grale and Dynamic GUS, therefore we plotted only one line for each dataset.}
\end{figure*}

\begin{figure*}
    \centering
    \begin{subfigure}[b]{0.48\textwidth}
        \centering
        \includegraphics[width=\textwidth]{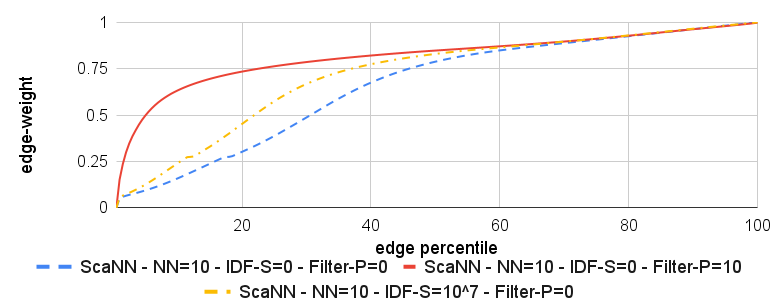}
        \caption{ogbn-products}
        \label{sfig:gus-sketch-nn10-prod}
    \end{subfigure}
    \hfill
    \begin{subfigure}[b]{0.48\textwidth}
        \centering
        \includegraphics[width=\textwidth]{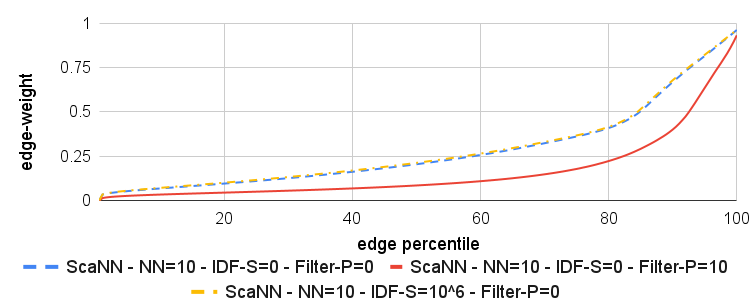}
        \caption{ogbn-arxiv}
        \label{sfig:gus-sketch-nn10-arxiv}
    \end{subfigure}
    \begin{subfigure}[b]{0.48\textwidth}
        \centering
        \includegraphics[width=\textwidth]{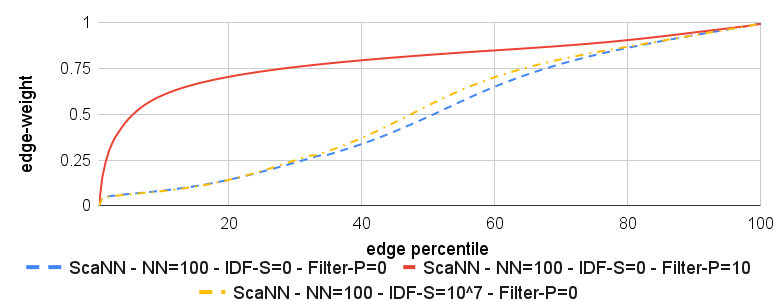}
        \caption{ogbn-products}
        \label{sfig:gus-sketch-nn100-prod}
    \end{subfigure}
    \hfill
    \begin{subfigure}[b]{0.48\textwidth}
        \centering
        \includegraphics[width=\textwidth]{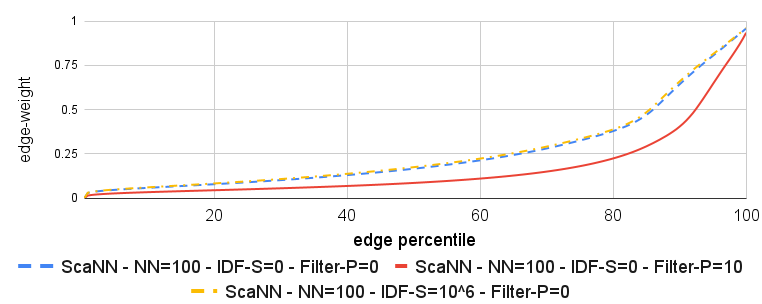}
        \caption{ogbn-arxiv}
        \label{sfig:gus-sketch-nn100-arxiv}
    \end{subfigure}
    \begin{subfigure}[b]{0.48\textwidth}
        \centering
        \includegraphics[width=\textwidth]{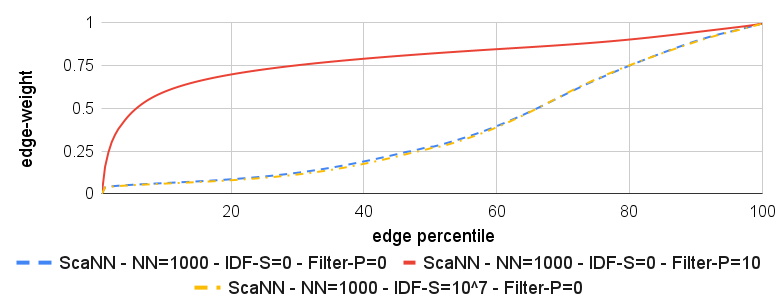}
        \caption{ogbn-products}
        \label{sfig:gus-sketch-nn1000-prod}
    \end{subfigure}
    \hfill
    \begin{subfigure}[b]{0.48\textwidth}
        \centering
        \includegraphics[width=\textwidth]{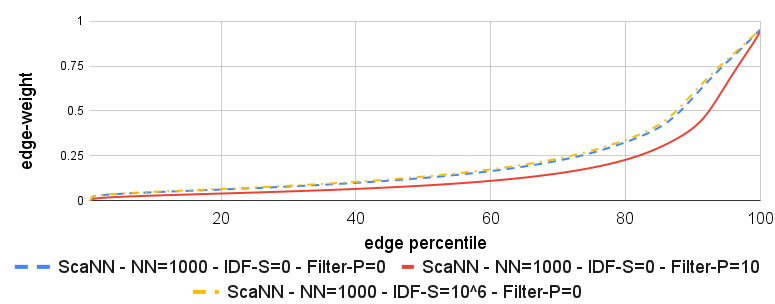}
        \caption{ogbn-arxiv}
        \label{sfig:gus-sketch-nn1000-arxiv}
    \end{subfigure}
\caption{\label{fig:gus-edges-nn=10} The edge-weight distribution when edges are retrieved using Dynamic GUS with (\subref{sfig:gus-sketch-nn10-prod}-\subref{sfig:gus-sketch-nn10-arxiv}) ScaNN-NN=10, (\subref{sfig:gus-sketch-nn100-prod}-\subref{sfig:gus-sketch-nn100-arxiv}) ScaNN-NN=100, (\subref{sfig:gus-sketch-nn1000-prod}-\subref{sfig:gus-sketch-nn1000-arxiv}) ScaNN-NN=1000, for varying IDF and filtering parameters. The total number of edges retrieved for different values of Filter-P, independently of the other parameters, are: 
(\subref{sfig:gus-sketch-nn10-prod}) $22,041,932$ when Filter-P=0 and $19,020,281$ when Filter-P=10;
(\subref{sfig:gus-sketch-nn10-arxiv}) $1,524,087$ when Filter-P=0 and $1,416,463$ when Filter-P=10;
(\subref{sfig:gus-sketch-nn100-prod}) $242,454,448$ when Filter-P=0 and $103,396,128$ when Filter-P=10;
(\subref{sfig:gus-sketch-nn100-arxiv}) $16,764,957$ when Filter-P=0 and $15,442,048$ when Filter-P=10;
(\subref{sfig:gus-sketch-nn1000-prod}) $2,446,579,971$ when Filter-P=0 and $135,807,863$ when Filter-P=10;
(\subref{sfig:gus-sketch-nn1000-arxiv}) $169,173,657$ when Filter-P=0 and $119,986,924$ when Filter-P=10.}
\end{figure*}

\begin{figure*}
    \centering
    \begin{subfigure}[b]{0.45\textwidth}
        \centering
        \includegraphics[width=\textwidth]{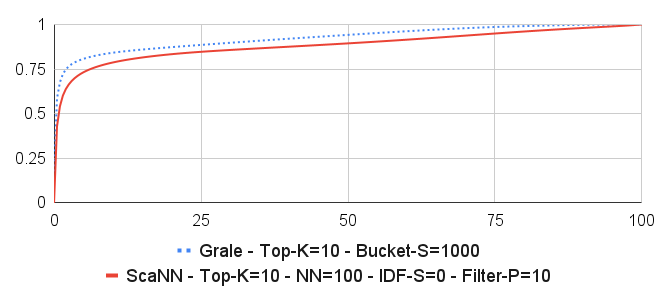}
        \caption{ogbn-products}
    \end{subfigure}
    \hfill
    \begin{subfigure}[b]{0.45\textwidth}
        \centering
        \includegraphics[width=\textwidth]{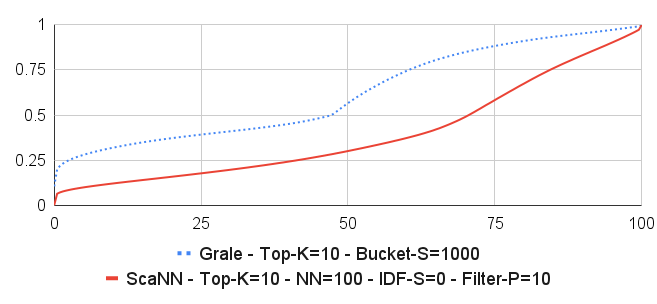}
        \caption{ogbn-arxiv}
    \end{subfigure}
\caption{\label{fig:with-pruning-nn=10} The edge-weight distribution when edges are retrieved using Grale with Top-K=10 and Bucket-S=1000, and GUS with ScaNN-NN=10  and the best-performing values for IDF-S and Filter-P. The total number of edges retrieved for ogbn-products by Grale -- Top-K=10 --  Bucket-S=1000 is $24,490,290$ and by ScaNN -- Top-K=10 -- NN=100 -- IDF-S=0 -- Filter-P=10 is $20,866,944$. For ogbn-arxiv the respective numbers are $1,693,430$ and $1,572,547$.}
\end{figure*}

In this section we compare the empirical quality of Grale and Dynamic GUS and also study the performance of the latter in a dynamic environment. More precisely:

\begin{enumerate}
    \item We compare the edges of the graph produced by Grale with the edges produced by an offline implementation of Dynamic Grale Using ScaNN (GUS) for a range of configurations. Note that the offline GUS and dynamic GUS provide identical results.%
    \item We provide latencies and resource usage of Dynamic GUS for different configurations in a dynamic setting.
\end{enumerate}

In this work, we focus on two datasets, both available at the Open Graph Benchmark.\footnote{https://ogb.stanford.edu/docs/nodeprop/}

\begin{itemize}
    \item ogbn-arxiv contains $169,343$ arXiv papers in Computer Science (CS) indexed by MAG \cite{DBLP:journals/qss/WangSHWDK20}. This dataset has two features for each paper: i) the publication year; and ii) a 128-dimensional feature vector obtained by averaging the embeddings of words present in its title and abstract. 
    \item ogbn-products contains $2,449,029$ Amazon products with two features for each product: i) a list of products that have been purchased together with it (based on the Amazon product co-purchasing network \cite{Bhatia16}); and ii) an embedding generated by extracting bag-of-words features from the product descriptions followed by a Principal Component Analysis to reduce the dimension to 100 \cite{DBLP:conf/kdd/ChiangLSLBH19}. 
\end{itemize}

\paragraph{Model training:} For each dataset, we train a model with the same architecture, which is a two-layer neural network with 10 hidden units per layer. The features explained above are used to train these models.

\paragraph{Bucket size for Grale:} Grale \cite{grale_paper} suggests a mechanism to bound the number of points in each bucket. As explained in that paper, Grale computes a set of buckets for each point and then, for each bucket, computes the set of points that belong to it. This might result in the creation of buckets with a large number of points, which has a negative impact on the running time. To mitigate that, the paper suggests defining a maximum bucket size, and randomly subdividing any buckets that are larger than that size. In this section, we run experiments on Grale with and without that bucket subdivision step. 

\subsection{Offline Experiments}
\label{sec:offline-experiments}

\paragraph{Plots.} For all plots related to this section (\cref{fig:all-edges,fig:grale-edges,fig:gus-edges-nn=10,fig:grale-vs-gus-nn=10,fig:with-pruning-nn=10,fig:with-pruning-nn=100}), we plot the edge-weight at different percentiles of the edges ordered by edge weight. To properly compare two algorithms using this plot, one should take into account the total number of edges produced by each algorithm. If two algorithms produce the same total number of edges, the one that appears above on the plot for a fixed percentile performs better at the specific percentile.

\paragraph{First experiment.} In the first experiment we compare the edges without setting a limit on bucket size in Grale, and retrieving all the edges with negative distance in ScaNN. The goal of this experiment is to experimentally validate \cref{lemma:same-output}. We present in \cref{fig:all-edges} the edge-weight for the datasets ogbn-arxiv and ogbn-products. We observe that Dynamic GUS achieves identical result as Grale in this setting as claimed by the lemma.

\paragraph{Second experiment} In the second experiment we study the effects of the number of nearest neighbors retrieved from ScaNN (denoted by ScaNN-NN), the percentage of overly popular buckets filtered (denoted by Filter-P), and the size of the inverse document frequency table (IDF-S). Filter-P equal to $x$ means that $x\%$ of the buckets that have the highest cardinalities are ignored. Moreover, IDF-S equal to $x$ means that we compute $x$ buckets with highest inverse document frequency and the weight of the rest is equal to the $x$-th highest weight. The reason behind computing only the top $x$ weights is to ensure that the size of the table is bounded and can easily be stored in memory. Furthermore, for Grale we run experiments for different bucket split sizes (Bucket-S). More precisely, if Bucket-S is set to $m$, we randomly divide each bucket that has more than $m$ points to ensure that the maximum size of each bucket is at most $m$. Note that since Grale does not maintain a spatial representation of the points apart from the sketch buckets produced by LSH, the number of edges returned for a point $p$ is always the number of scoring pairs that contain $p$. Therefore, the computational cost does not decrease if the user needs much fewer neighbors than output by Grale. The results are presented in \cref{fig:gus-edges-nn=10} and \cref{fig:grale-edges} (in the appendix). We observe that:
\begin{itemize}
    \item The edges produced by Dynamic GUS have higher quality (since their weights are higher when scored by the model) compared to Grale. This results from the more effective way in which Dynamic GUS limits the number of comparisons: while Grale does that by randomly splitting buckets, Dynamic GUS retrieves the set of the closest points to the query point in the embedding space, which are likely to have many bucket IDs (or some bucket IDs with high IDF weights) in common with the query point. For example, in \cref{fig:grale-edges-products} we observe that less than $30\%$ of the edges retrieved by Grale have a weight larger than $0.25$ (for all bucket sizes) for the ogbn-products dataset. In comparison, more than $97\%$ of the edges retrieved by Dynamic GUS with Filter-P$=10\%$ have weight above $0.25$, as shown in \cref{fig:gus-edges-nn=10}. For ease of comparison between the edge weight distribution of Grale and Dynamic GUS, we present the same results in a different format in \cref{fig:grale-vs-gus-nn=10} (in the appendix).
    \item From the same plots, we observe that filtering overly popular bucket IDs and using inverse document frequency improve the quality of the edges of the graph. For instance, in \cref{sfig:gus-sketch-nn10-prod} the edge weight at $20$ percentile is almost $0.3$ when filtering and IDF are not used, $0.45$ when IDF is used, and $0.74$ when filtering is used.
\end{itemize}

\paragraph{Third Experiment} Pruning the edges based on the weight assigned to them by the model is an important ingredient in most of the applications of Grale and Dynamic GUS at Google. Therefore, we also compare the quality of the graphs output by Dynamic GUS with the graphs output by Grale with such post-processing. For each point in the graph output by Grale, we keep the Top-$K$ neighbors. In particular, we filter the top 100 and 1000 neighbors, where we set ScaNN-NN equal to $K$ for Dynamic GUS. The results are shown in \cref{fig:with-pruning-nn=10} and \cref{fig:with-pruning-nn=100} (in the appendix). We observe that Grale and Dynamic GUS have high and comparable edge weights for ogbn-products and Top-K$=10$. For Top-K$=100$, the quality of Dynamic GUS is significantly higher, despite the fact that the Grale algorithm scored a factor $2-5$ more edges than Dynamic GUS. For ogbn-arXiv, we observe that the scores of the edges found by Dynamic GUS are slightly lower. Recall that the computational cost of Grale is not decreased by the Top-K post-processing, which is one of the advantages of Dynamic GUS. In the next section, we show how choosing ScaNN-NN can have a significant improvement on the performance with respect to latency.

\subsection{Dynamic Experiments}

In this section we focus on the performance of Dynamic GUS in a dynamic setting. The experiments were run on a single Google Cloud machine (machine type n2d-standard-48 with CPU platform AMD EPYC~7B12). All computation tasks of each experiment were run on a single core (vCPU), one experiment at a time, and queries were sent sequentially, one by one, to Dynamic GUS (running on the same core). For each dataset, we sampled a random subset of 10,000 points. For each combination of ScaNN-NN, IDF-S and Filter-P, we queried the neighbors of each of those points. We analyzed the query latency (wall clock time from sending the request to receiving the response) in \cref{fig:latency} (in the appendix), and the average CPU time per request (as the total CPU time of the respective experiment divided by the number of queries) and the maximum memory usage of the respective experiment in \cref{tbl:resources} (in the appendix). We do not replicate the quality analysis here, as that has already been done in \cref{sec:offline-experiments} using an offline implementation of Dynamic GUS that produces identical results. We also note that all operations are run sequentially and centralized only for the interpretability and stability of the experiments, and the algorithm can be run in a parallel and distributed setting for larger datasets. We observe that for the ogbn-arxiv dataset, the median latency is around 10 to 20 milliseconds. Similarly, the median latency for ogbn-products is in the range between 5 and 25 milliseconds. This experiment shows that the latency is indeed in order of tens of milliseconds, which is consistent with the latencies observed across use cases at Google. We also observe that choosing smaller values of ScaNN-NN has significant impact on the median latencies and latencies of outliers for both datasets. The same holds for larger values of Filter-P, which suggests that reducing the dimension of the  embedding space has positive impact on the performance as well.

Furthermore, we also measured the median wall clock time for insertions of points into the datasets, which is 0.29ms for ogbn-arxiv (95\%ile = 0.54ms) and 0.42ms for ogbn-products (95\%ile = 0.78ms).

\bibliography{references}
\clearpage

\appendix
\section{Additional Experimental Evaluations}
\label{sec:app_add_exp}
\begin{figure*}[ht]
    \centering
    \begin{subfigure}[b]{0.45\textwidth}
        \centering
        \includegraphics[width=\textwidth]{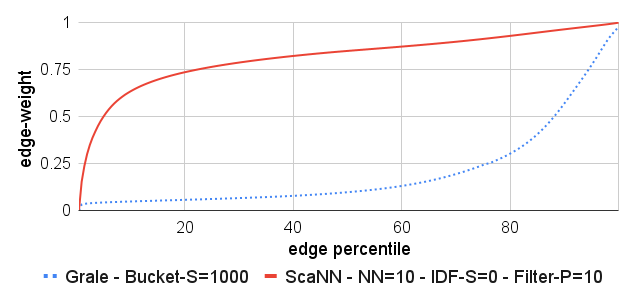}
        \caption{ogbn-products}
        \label{sfig:grale-gus-nn10-prod}
    \end{subfigure}
    \hfill
    \begin{subfigure}[b]{0.45\textwidth}
        \centering
        \includegraphics[width=\textwidth]{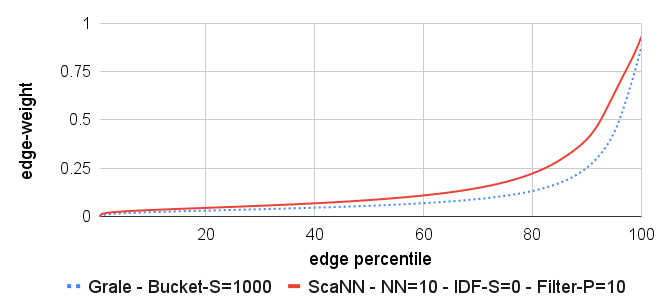}
        \caption{ogbn-arxiv}
        \label{sfig:grale-gus-nn10-arxiv}
    \end{subfigure}
    \begin{subfigure}[b]{0.45\textwidth}
        \centering
        \includegraphics[width=\textwidth]{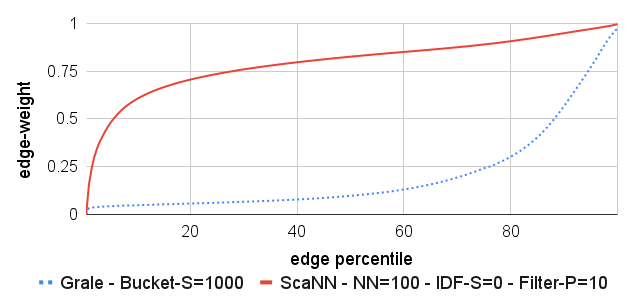}
        \caption{ogbn-products}
        \label{sfig:grale-gus-nn100-prod}
    \end{subfigure}
    \hfill
    \begin{subfigure}[b]{0.45\textwidth}
        \centering
        \includegraphics[width=\textwidth]{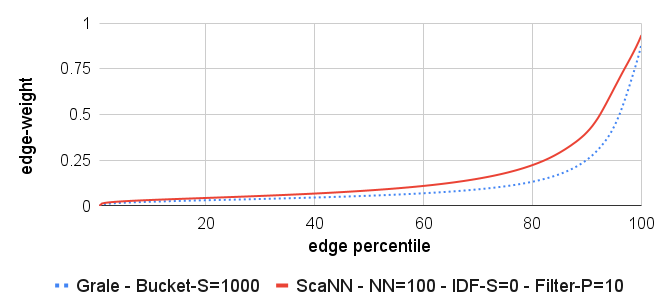}
        \caption{ogbn-arxiv}
        \label{sfig:grale-gus-nn100-arxiv}
    \end{subfigure}
    \begin{subfigure}[b]{0.45\textwidth}
        \centering
        \includegraphics[width=\textwidth]{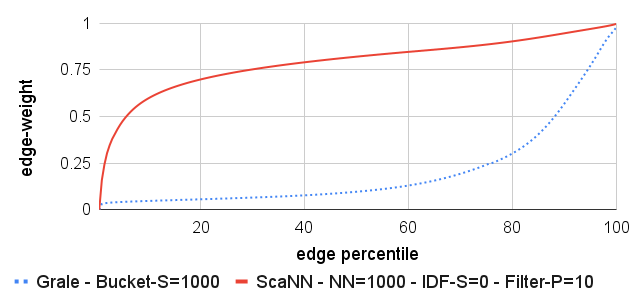}
        \caption{ogbn-products}
        \label{sfig:grale-gus-nn1000-prod}
    \end{subfigure}
    \hfill
    \begin{subfigure}[b]{0.45\textwidth}
        \centering
        \includegraphics[width=\textwidth]{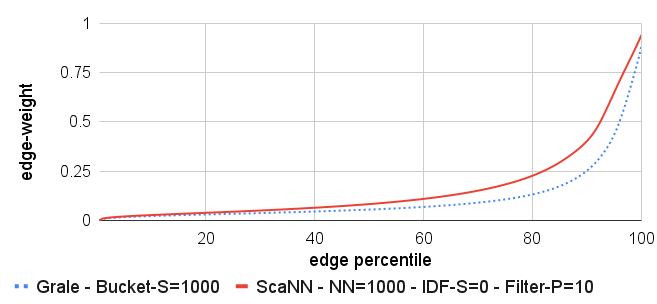}
        \caption{ogbn-arxiv}
        \label{sfig:grale-gus-nn1000-arxiv}
    \end{subfigure}
\caption{\label{fig:grale-vs-gus-nn=10} The edge-weight distribution when edges are retrieved using Grale with Bucket-S=1000, and GUS with (\subref{sfig:grale-gus-nn10-prod}-\subref{sfig:grale-gus-nn10-arxiv}) ScaNN-NN=10, (\subref{sfig:grale-gus-nn100-prod}-\subref{sfig:grale-gus-nn100-arxiv}) ScaNN-NN=100, (\subref{sfig:grale-gus-nn1000-prod}-\subref{sfig:grale-gus-nn1000-arxiv}) ScaNN-NN=1000 and the best-performing values for IDF-S=0 and Filter-P=10. The total numbers of edges retrieved are 
(\subref{sfig:grale-gus-nn10-prod})  $140,242,620,100$ by Grale and $19,020,281$ by ScaNN;
(\subref{sfig:grale-gus-nn10-arxiv}) $11,028,453,890$ by Grale and $1,416,463$ by ScaNN; 
(\subref{sfig:grale-gus-nn100-prod}) $140,242,620,100$ by Grale and $103,396,128$ by ScaNN; 
(\subref{sfig:grale-gus-nn100-arxiv}) $11,028,453,890$ by Grale and $15,442,048$ by ScaNN;
(\subref{sfig:grale-gus-nn1000-prod}) $140,242,620,100$ by Grale $135,807,863$ by ScaNN; 
(\subref{sfig:grale-gus-nn1000-arxiv}) $11,028,453,890$ by Grale and $119,986,924$ by ScaNN.}
\end{figure*}

\begin{figure*}
    \centering
    \begin{subfigure}[b]{0.48\textwidth}
        \centering
        \includegraphics[width=\textwidth]{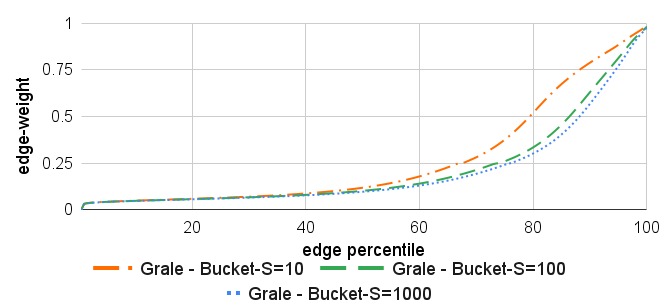}
        \caption{ogbn-products}
        \label{fig:grale-edges-products}
    \end{subfigure}
    \hfill
    \begin{subfigure}[b]{0.48\textwidth}
        \centering
        \includegraphics[width=\textwidth]{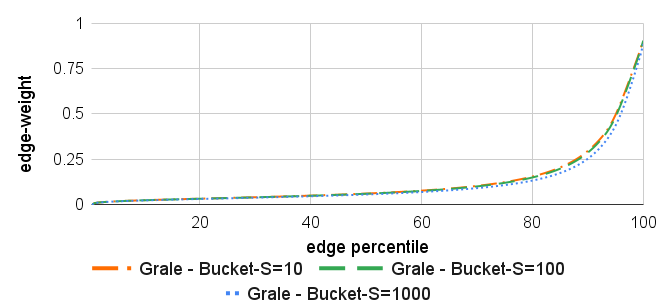}
        \caption{ogbn-arxiv}
    \end{subfigure}
\caption{\label{fig:grale-edges} The edge-weight distribution of the edges retrieved using Grale with varying sizes of buckets for ogbn-arxiv and ogbn-product. The number of edges retrieved for ogbn-products are $2,418,660,128$ for Grale -- Bucket-S=10, $22,806,807,984$ for Grale -- Bucket-S=100, and $140,242,620,100$ for Grale -- Bucket-S=1000. For ogbn-arxiv the respective numbers are $152,316,856$, $1,577,041,422$, and $11,028,453,890$.}
\end{figure*}

\begin{figure*}
    \centering
    \begin{subfigure}[b]{0.45\textwidth}
        \centering
        \includegraphics[width=\textwidth]{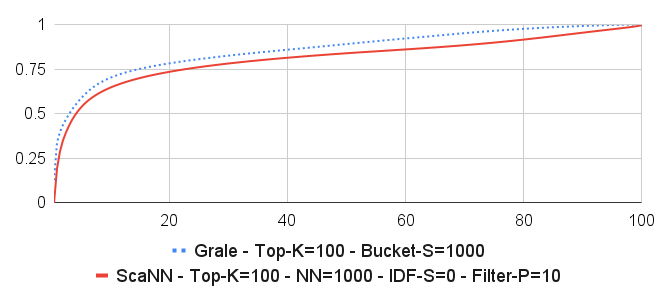}
        \caption{ogbn-products}
    \end{subfigure}
    \hfill
    \begin{subfigure}[b]{0.45\textwidth}
        \centering
        \includegraphics[width=\textwidth]{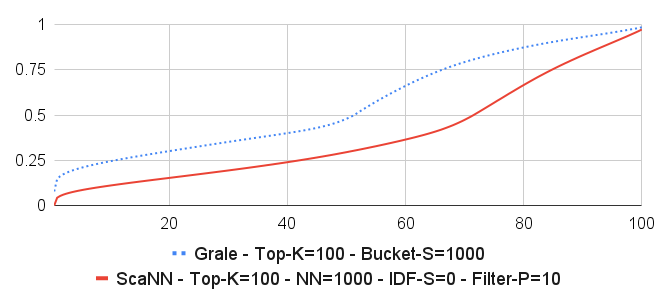}
        \caption{ogbn-arxiv}
    \end{subfigure}
\caption{\label{fig:with-pruning-nn=100} The edge-weight distribution when edges are retrieved using Grale with Top=K=100 and Bucket-S=1000, and GUS with ScaNN-NN=100 and the best-performing values for IDF-S and Filter-P. The total number of edges retrieved for ogbn-products by Grale -- Top-K=100 --  Bucket-S=1000 is $244,902,900$ and by ScaNN -- Top-K=100 -- NN=1000 -- IDF-S=0 -- Filter-P=10 is $103,768,901$. For ogbn-arxiv the respective numbers are $16,934,300$ and $15,596,017$.}
\end{figure*}

\begin{figure*}
    \centering
    \begin{subfigure}[b]{0.49\textwidth}
        \centering
        \includegraphics[width=\textwidth]{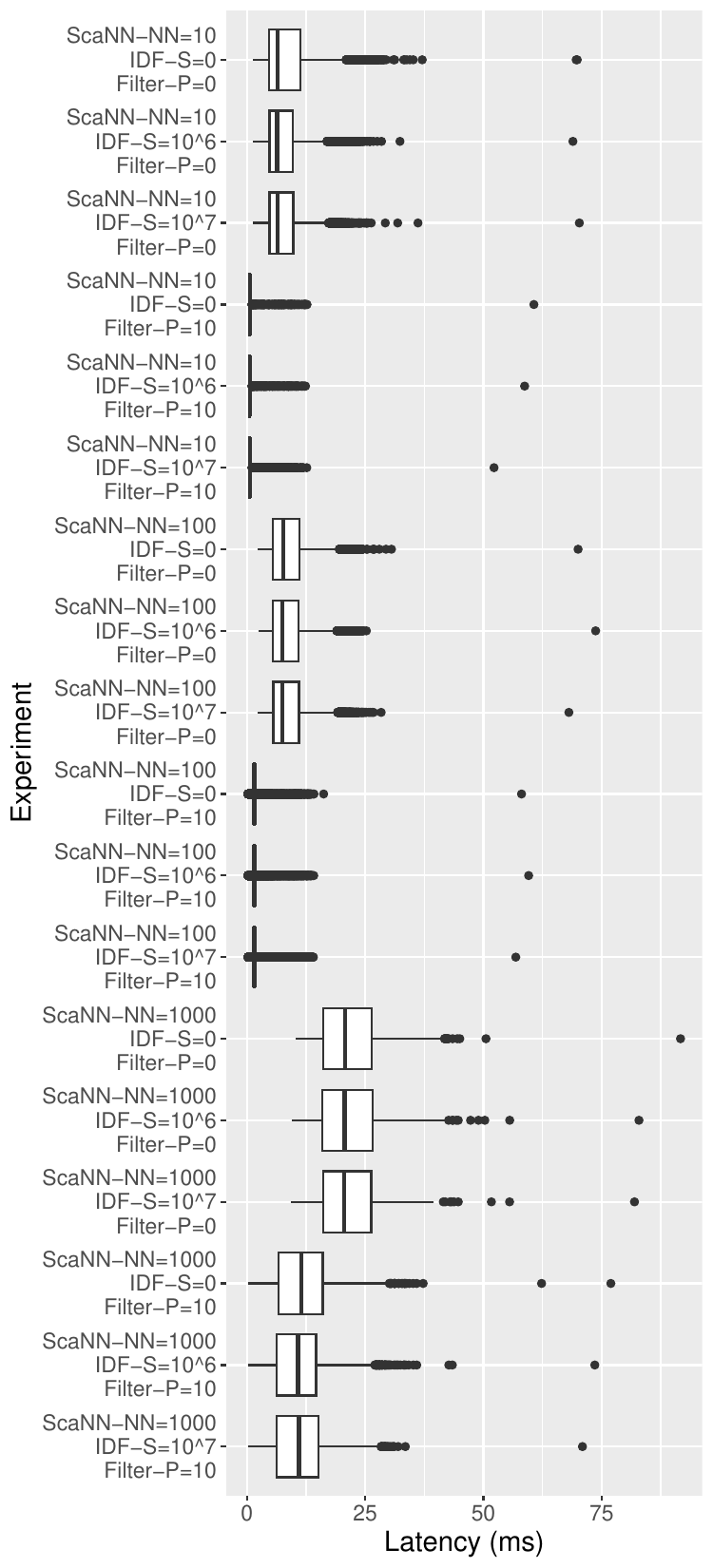}
        \caption{ogbn-arXiv}
    \end{subfigure}
    \hfill
    \begin{subfigure}[b]{0.49\textwidth}
        \centering
        \includegraphics[width=\textwidth]{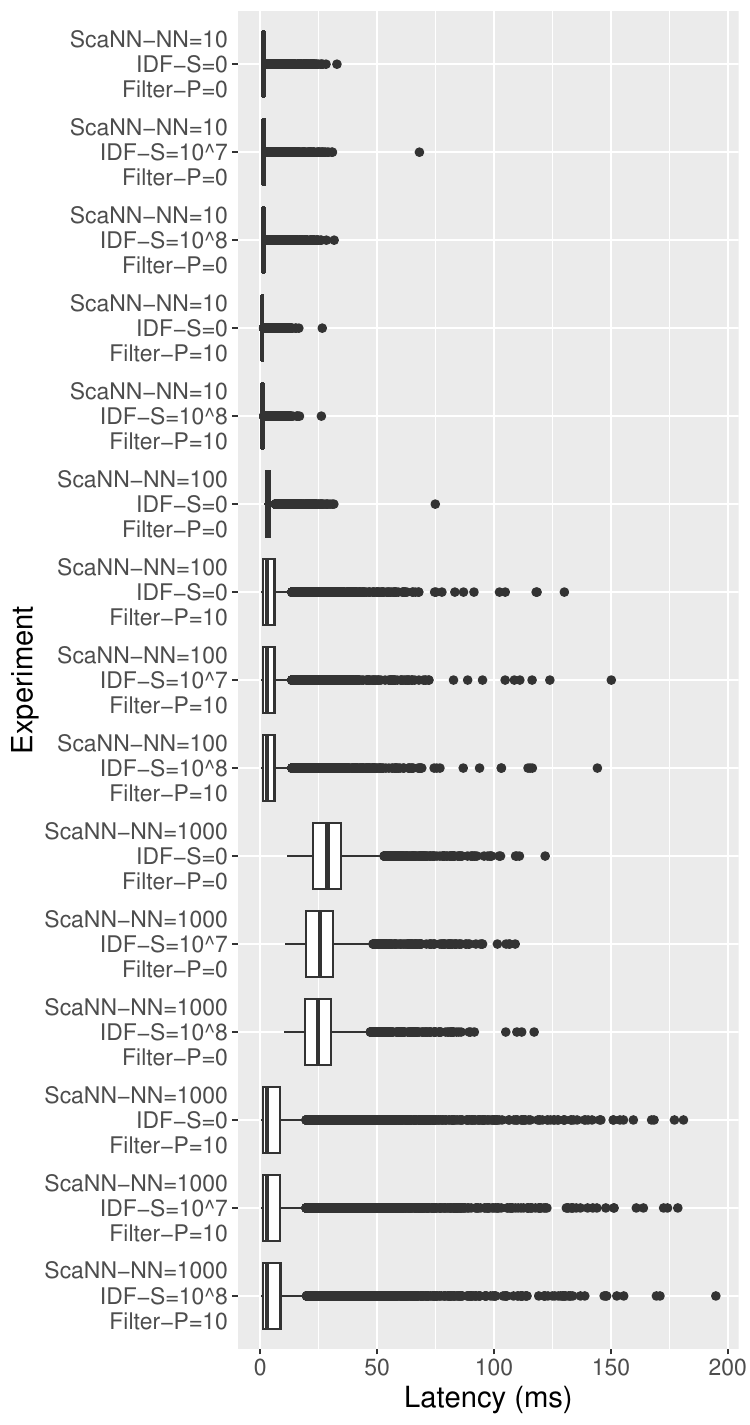}
        \caption{ogbn-products}
    \end{subfigure}
    \caption{\label{fig:latency} Latency distribution for the datasets ogbn-arxiv and ogbn-products in the online experiment. ScaNN-NN denotes the number of neighbors retrieved from ScaNN; IDF-S denotes the number of inverse document frequencies computed; and Filter-P denotes the percentage of overly popular bucket IDs filtered.}
\end{figure*}

\begin{figure*}
    \centering
    \subcaptionbox{ogbn-arXiv}{
\begin{tabular}{lllll}
  \toprule
ScaNN-NN & IDF-S & Filter-P & Avg. time & Max. mem. \\ 
  \midrule
10 & 0 & 0 & 14 ms & 675 MiB \\ 
  10 & 10\verb|^|6 & 0 & 14 ms & 688 MiB \\ 
  10 & 10\verb|^|7 & 0 & 14 ms & 677 MiB \\ 
  10 & 0 & 10 & 14 ms & 501 MiB \\ 
  10 & 10\verb|^|6 & 10 & 14 ms & 502 MiB \\ 
  10 & 10\verb|^|7 & 10 & 14 ms & 503 MiB \\ 
  100 & 0 & 0 & 14 ms & 675 MiB \\ 
  100 & 10\verb|^|6 & 0 & 14 ms & 679 MiB \\ 
  100 & 10\verb|^|7 & 0 & 14 ms & 686 MiB \\ 
  100 & 0 & 10 & 14 ms & 504 MiB \\ 
  100 & 10\verb|^|6 & 10 & 14 ms & 504 MiB \\ 
  100 & 10\verb|^|7 & 10 & 14 ms & 516 MiB \\ 
  1000 & 0 & 0 & 14 ms & 681 MiB \\ 
  1000 & 10\verb|^|6 & 0 & 14 ms & 682 MiB \\ 
  1000 & 10\verb|^|7 & 0 & 14 ms & 678 MiB \\ 
  1000 & 0 & 10 & 14 ms & 504 MiB \\ 
  1000 & 10\verb|^|6 & 10 & 14 ms & 502 MiB \\ 
  1000 & 10\verb|^|7 & 10 & 14 ms & 503 MiB \\ 
   \bottomrule
\end{tabular}
    }
    \hfill
    \subcaptionbox{ogbn-products}{
        \begin{tabular}{lllll}
  \toprule
ScaNN-NN & IDF-S & Filter-P & Avg. time & Max. mem. \\ 
  \midrule
10 & 0 & 0 & 33 ms & 16684 MiB \\ 
  10 & 10\verb|^|7 & 0 & 33 ms & 16953 MiB \\ 
  10 & 10\verb|^|8 & 0 & 33 ms & 18868 MiB \\ 
  10 & 0 & 10 & 33 ms & 13086 MiB \\ 
  10 & 10\verb|^|7 & 10 & 33 ms & 26645 MiB \\ 
  10 & 10\verb|^|8 & 10 & 33 ms & 14973 MiB \\ 
  100 & 0 & 0 & 33 ms & 16678 MiB \\ 
  100 & 10\verb|^|7 & 0 & 33 ms & 33871 MiB \\ 
  100 & 10\verb|^|8 & 0 & 33 ms & 18854 MiB \\ 
  100 & 0 & 10 & 33 ms & 13133 MiB \\ 
  100 & 10\verb|^|7 & 10 & 33 ms & 13365 MiB \\ 
  100 & 10\verb|^|8 & 10 & 33 ms & 14981 MiB \\ 
  1000 & 0 & 0 & 33 ms & 16685 MiB \\ 
  1000 & 10\verb|^|7 & 0 & 33 ms & 16959 MiB \\ 
  1000 & 10\verb|^|8 & 0 & 33 ms & 18867 MiB \\ 
  1000 & 0 & 10 & 33 ms & 13141 MiB \\ 
  1000 & 10\verb|^|7 & 10 & 33 ms & 13355 MiB \\ 
  1000 & 10\verb|^|8 & 10 & 33 ms & 14982 MiB \\ 
   \bottomrule
\end{tabular}

    }
    \caption{\label{tbl:resources} The average (end-to-end) CPU time per query and maximum memory (in MiB) usage of each online experiment. ScaNN-NN denotes the number of neighbors retrieved from ScaNN; IDF-S denotes the number of inverse document frequencies computed; and Filter-P denotes the percentage of overly popular bucket IDs filtered.}
\end{figure*}

\end{document}